\documentclass{article}

\usepackage[a4paper,top=2.2cm,bottom=2.6cm,left=2.2cm,right=2.2cm,heightrounded,includefoot]{geometry}

\usepackage{amsmath, amssymb, braket, amsthm}
\usepackage{tikz,hyperref,cleveref}
\usepackage[vlined,linesnumbered]{algorithm2e}
\usepackage{authblk}

\hypersetup{
  colorlinks=true,
  linkcolor=blue,
  citecolor=teal,
  urlcolor=blue
}

\usepackage{thmtools}

\declaretheorem[style=plain]{theorem}
\declaretheorem[style=plain,numberlike=theorem]{lemma,corollary,proposition,claim}

\title{A Space-space Trade-off for Directed st-Connectivity}
\author{Roman Edenhofer}
\affil{Université Paris Cité, CNRS, IRIF, Paris, France}
\date{}

\begin{document}

\maketitle

\begin{abstract}
    We prove a space-space trade-off for directed $st$-connectivity in the catalytic space model.
    For any integer $k \le n$, we give an algorithm that decides directed $st$-connectivity using~$O(\log n \cdot \log k+\log n)$ regular workspace and~$O\!\left(\frac{n}{k} \cdot \log^2 n\right)$ bits of catalytic memory.
    This interpolates between the classical~$O(\log^2 n)$-space bound from Savitch's algorithm
    and a catalytic endpoint with~$O(\log n)$ workspace and $O(n\cdot \log^2 n)$ catalytic memory.

    As a warm-up, we present a catalytic variant of Savitch's algorithm achieving the endpoint above.
    Up to logarithmic factors, this matches the smallest catalyst size currently known for catalytic logspace algorithms, due to Cook and Pyne (ITCS~2026).
    Our techniques also extend to counting the number of walks from $s$ to $t$ of a given length $\ell\leq n$.
\end{abstract}

\section{Introduction}

Directed $st$-connectivity, denoted $\mathrm{STCON}$, is the decision problem where we are given a directed graph $G$ on~$n$ vertices and two designated vertices $s,t\in V(G)$, and where we are asked whether there exists a directed path from $s$ to $t$.
The problem is of fundamental importance in space-bounded complexity theory because it is complete for \emph{nondeterministic logspace}, $\mathsf{NL}$.

The best-known deterministic strategy for solving $\mathrm{STCON}$ in small space is Savitch's algorithm~\cite{Sav70}, which uses only $O(\log^2 n)$ bits of workspace.
Improving on Savitch's algorithm, or proving its optimality, would constitute a major breakthrough in our understanding of space-bounded computation.

The problem has also been studied in the \emph{catalytic logspace} ($\mathsf{CL}$) model, where an algorithm has $O(\log n)$ bits of regular workspace and, in addition, access to a polynomial-size auxiliary bitstring, the \emph{catalyst}, that may be used during the computation but must be restored to its initial content when the computation halts.
The model was introduced by Buhrman, Cleve, Kouck\'y, Loff and Speelman \cite{BCKLS14} and is surprisingly powerful. 
In the original paper, the authors already observed that catalytic logspace can simulate logarithmic-depth threshold $\mathsf{TC}^1$-circuits.
Since $\mathsf{TC}^1$ contains $\mathsf{NL}$, and in fact even the counting class $\mathsf{\#L}$, it is long-known that~$\mathrm{STCON}$ can be solved in catalytic logspace.
A natural question however is, how large the catalyst must be in order to solve $\mathrm{STCON}$ (and thus to simulate $\mathsf{NL}$ computations).

Recently, Cook and Pyne~\cite{CP25} gave a catalytic logspace algorithm for~$\mathrm{STCON}$ using a catalyst of size~$\widetilde{O}(n)$.
They also remarked that an algorithm with sublinear catalyst would imply a randomized polynomial-time, sublinear-space algorithm for $\mathrm{STCON}$.
While Barnes, Buss, Ruzzo, and Schieber~\cite{BBRS98} gave a polynomial-time algorithm using space $O(n/2^{\Theta(\sqrt{\log n})})$, it remains a long-standing open problem to obtain a polynomial-time algorithm running in truly sublinear space, i.e.\ $O(n^{1-\varepsilon})$ for some $\varepsilon>0$.
Some researchers have even conjectured that this is impossible (see e.g.~\cite{Wig92}).

In this work, we investigate a trade-off between the amount of regular workspace and the amount of catalytic memory needed to solve $\mathrm{STCON}$.
We show that one can interpolate smoothly between Savitch's classical $O(\log^2 n)$ space bound and a catalytic endpoint with $O(\log n)$ workspace and $O(n\cdot \log^2 n)$ bits of catalytic memory, matching the catalyst size of Cook and Pyne up to logarithmic factors.

Formally, our main result is the following.

\begin{restatable}{theorem}{STCONTradeoff}\label{thm: STCON Space-space tradeoff}
    Given a directed graph $G$ of size $n$, vertices $s,t\in V(G)$, a length $\ell\in[n]$ and a parameter~$k\in[n]$, there is an algorithm running in $\mathsf{CSPACE}(O(\log \ell\cdot\log k+\log n),\;O(\frac{n}{k}\cdot\log^2 n))$ that outputs the number of length-$\ell$ walks from $s$ to $t$.
    In particular, for any $k\in [n]$,
    \[
        \mathrm{STCON}\in \mathsf{CSPACE}\left(O\left(\log n\cdot\log k+\log n\right),\;O\left(\frac{n}{k}\cdot\log^2 n\right)\right).
    \]
\end{restatable}

\noindent Here, $\mathsf{CSPACE}(S,C)$ denotes catalytic computation with $S$ bits of regular workspace and $C$ bits of catalytic memory.

Independently from us, Chmel, Dudeja, Koucký, Mertz and Rajgopal \cite{CDKMR26} came up with essentially the same algorithm which they use as a subroutine for short length $\ell=2^{O(\sqrt{\log n})}$ $\mathrm{STCON}$-instances and with the choice $k=2^{O(\sqrt{\log n})}$, together with a clever vertex-sparsification of the input graph, to put general~$\mathrm{STCON}$ in~$\mathsf{CSPACE}(O(\log n),\; n/2^{\Theta(\sqrt{\log n})})$.

Since any nondeterministic computation using $S(n)$ bits of workspace induces a configuration graph of size $O(2^{S})$, we obtain the following corollary from our main result.

\begin{corollary}
For all $1\le K\le S$,
\[
\mathsf{NSPACE}(S)\subseteq \mathsf{CSPACE}\!\left(O(S\cdot K),\; O\!\left(2^{S-K}\cdot S^2\right)\right).
\]
\end{corollary}

\subsection{Catalytic Space and Register Programs}

\noindent\textbf{Catalytic Turing machines.}
A \emph{catalytic Turing machine} is a multi-tape Turing machine with a read-only input tape, a write-only output tape, a (regular) read-and-write work tape, and an additional read-and-write \emph{catalytic} tape.
The catalytic tape is initialized with an auxiliary bitstring that may be modified during the computation, under the promise that the machine must restore the catalytic tape to its initial contents upon halting.

We say that a catalytic Turing machine $M$ runs in $\mathsf{CSPACE}(S(n),C(n))$ if, for every input $x\in\{0,1\}^n$, the computation of $M$ on $x$ visits at most $S(n)$ distinct cells on the work tape and at most $C(n)$ distinct cells on the catalytic tape.
For a language $L\subseteq\{0,1\}^*$, we write $L\in \mathsf{CSPACE}(S(n),C(n))$ if there exists a catalytic Turing machine deciding $L$ within these bounds.
We define $\mathsf{CL}$ to be the class of languages decidable in $\mathsf{CSPACE}(O(\log n),\mathrm{poly}(n))$.

The catalytic space model has proven to be very useful in enhancing our understanding of space-bounded computation. Most notably, using catalytic techniques, Cook and Mertz \cite{CM24} showed that the tree evaluation problem can be solved in near-logarithmic $O(\log n \cdot \log \log n)$ space, a problem which was originally introduced by Cook, McKenzie, Wehr, Braverman and Santhanam~\cite{CMMBS12} as a candidate to potentially separate $\mathsf{L}$ and $\mathsf{P}$.
Cook and Mertz's finding was recently used by Williams \cite{Will25} to prove that any time $t$ computation can be simulated in space $O(\sqrt{t\log t})$.

\medskip
\noindent\textbf{Catalytically Simulating Register Programs.}
Following much of the literature on catalytic computation, it is convenient to view the catalytic tape as a collection of registers $R_1,\dots,R_m$.
We interpret each register as storing a value in $\mathbb{Z}_q$ (in binary), and we will implement reversible arithmetic modulo $q$ for a suitable integer $q$.
In particular, we will need reversible updates of the form
\[
    R_i \leftarrow R_i \pm R_j \pmod{q}
    \qquad\text{and}\qquad
    R_i \leftarrow R_i \pm c \pmod{q},
\]
where $c\in \mathbb{Z}_q$ is a constant.

If $q$ is a power of $2$, these updates are straightforward to implement with $O(\log n)$ bits of classical control:
we allocate exactly $\log q$ catalytic bits per register, so that the register contents always represent an integer in $\{0,1,\dots,q-1\}$, and the total catalytic space is $m\log q$.

In our applications, however, we will want to choose moduli $q$ that are \emph{not} powers of $2$.
This creates an immediate encoding issue: if we simply allocate $\lceil \log q\rceil$ bits per register, then an arbitrary initial catalyst may encode integers in $\{0,1,\dots,2^{\lceil \log q\rceil}-1\}$, so many bitstrings might not correspond to valid residues modulo~$q$.

There are several standard ways to address this.
Buhrman, Cleve, Kouck\'y, Loff, and Speelman~\cite{BCKLS14} give a method that incurs a quadratic overhead in catalytic memory (which they already note is avoidable).
Cook and Pyne~\cite{CP25} provide a construction with only logarithmic-factor overhead: they enlarge each register slightly and apply a uniformly random shift so that, with some probability, the resulting contents encode a valid residue of some larger number.
Since our results do not require time efficiency, this randomness could also be removed by deterministically searching over all possible shifts.

In Appendix~\ref{appendix:A} we give a simple alternative solution that uses only $\lceil m\log q\rceil$ bits on the catalytic tape, which is the information-theoretic optimum, and might be of independent interest.
For the rest of the paper, we assume that every initial bitstring of a register always represents a valid integer.

\subsection{Our Techniques}

\noindent\textbf{A catalytic variant of Savitch's algorithm.}
We begin our paper by proving the catalytic endpoint of our trade-off (the case $k=1$ in \cref{thm: STCON Space-space tradeoff}), since it provides some of the main ideas used throughout the paper.
The starting point is a catalytic ``walk-propagation'' procedure, similar to one used by Cook and Pyne~\cite{CP25}, which propagates information along edges in a reversible manner.

Fix a modulus $q$.
For each ``layer'' $i\in[n]$ we maintain a block of registers $R_i[\cdot]$, indexed by vertices, where each register $R_i[v]\in \mathbb{Z}_q$ is stored in binary using $O(\log q)=O(\log n)$ catalytic bits.
Intuitively, the $i$th layer aggregates information about walks of length~$i-1$ ending at each vertex.
Given these blocks, we can propagate information forward by applying, for every directed edge $(u,v)\in E$, the reversible update
\begin{equation}\label{eqn: edge-push}
    R_i[v] \leftarrow R_i[v] + R_{i-1}[u] \pmod{q}.
\end{equation}
Performing these updates in a fixed increasing order for each $i\in\{2,...,n\}$ yields an easy to control reversible transformation of the catalyst.
After sweeping over all edges and layers, the register $R_n[t]$ contains aggregate information about the number of walks ending in $t$ of length $n-1$.
This information can be extracted by running the procedure multiple times with different initializations of $R_1[\cdot]$, as in~\cite{CP25}.

Our first new observation is that Savitch's divide-and-conquer idea can be used to \emph{shrink the catalyst size} required by such propagation.
Rather than maintaining all~$n$ layers simultaneously, it suffices to keep only~$\lceil\log n\rceil$ blocks of registers~$R_1[\cdot],\dots,R_{\lceil \log n\rceil}[\cdot]$ and propagate information in a recursive manner.
Concretely, instead of propagating length-$1$ information from layer $i-1$ to layer $i$ as in~\eqref{eqn: edge-push}, we design a procedure that propagates information corresponding to length-$\ell$ walks from one block of registers~$U[\cdot]$ to another block~$V[\cdot]$ by recursion on~$\ell$:
we first propagate length-$\ell/2$ information from $U$ into an intermediate block~$W$, and then propagate length-$\ell/2$ information from~$W$ into~$V$.
This yields a Savitch-style recursion depth of~$\lceil\log \ell\rceil$ while preserving reversibility and using only $O(\log n)$ regular workspace to control.

\medskip
\noindent\textbf{Our Space-space Trade-off.}
Our space-space trade-off builds on the catalytic Savitch-type propagation described above together with a key idea from~\cite{BBRS98}.

At a high level, Savitch's classical algorithm can be viewed as iterating space-efficiently over candidate $st$-walks $(v_1,\dots,v_n)\in V^n$ of length $n-1$ using $O(\log^2 n)$ space.
In~\cite{BBRS98}, the authors observed that one can replace this iteration over vertex sequences by an iteration over \emph{partition sequences}.
One of their subroutines partitions the vertex set into~$k$ groups and enumerates all length-$n$ sequences of parts~$(r_1,\dots,r_n)\in [k]^n$.
For a fixed partition sequence, one then verifies whether there exists an~$st$-walk whose~$i$th vertex lies in the part indexed by $r_i$.
The enumeration over partition sequences can again be carried out in a Savitch-style manner using $O(\log n\cdot \log k)$ space, while their verification procedure for a fixed sequence uses $O(n/k)$ additional (clean) workspace.

We show that this verification step can be performed using catalytic space instead of clean workspace with our catalytic Savitch variant, incurring only a polylogarithmic overhead.
Concretely, we implement the~$O(n/k)$-space verification using~$O\!\left(\frac{n}{k}\cdot \log^2 n\right)$ bits of catalytic memory (and $O(\log n)$ regular workspace for control),
while keeping the classical $O(\log n\cdot \log k)$-space Savitch-style enumeration of partition sequences unchanged.
Together, these ingredients yield the trade-off stated in~\cref{thm: STCON Space-space tradeoff}.

\medskip
\noindent\textbf{Organization.}
The organization of this paper is as follows:
In section \ref{sec: Catalytic Savitch} we will present and prove the correctness and complexity of our catalytic variant of Savitch's algorithm,
and in section \ref{sec: space-space trade-off} we will prove our main result, the space-space trade-off of \cref{thm: STCON Space-space tradeoff}.

\section{A Catalytic Variant of Savitch's Algorithm}\label{sec: Catalytic Savitch}

In this section we will prove the $k=1$ case of \cref{thm: STCON Space-space tradeoff} which we view as a catalytic-logspace variant of Savitch's algorithm.

Fix a prime modulus $q$ with $\log q = O(\log n)$.
Our catalytic algorithm uses $\lceil \log n\rceil+2$ blocks of $n$ registers each on the catalytic tape, denoted
\[
    U,\ V,\ W_1,\dots,W_{\lceil \log n\rceil}.
\]
For a block $R\in\{U,V,W_1,\dots,W_{\lceil \log n\rceil}\}$ and a vertex $v\in[n]$, we write $R[v]$ for the $v$th register in~$R$.
Each register stores an element of $\mathbb{Z}_q$ in binary, so the total catalyst size is $O(n\cdot \log^2 n)$.

\medskip
The main primitive of our algorithm is a family of reversible register programs $P^{(\ell)}(U,V,W_1,\dots,W_{\lceil \log \ell\rceil})$, which, informally, ``propagate'' information from block $U$ to block $V$ corresponding to walks of length~$\ell$.
We will make this interpretation precise later, for now, we focus on the definition and the key structural properties.

\medskip
\noindent\textbf{Base program.}
For $\ell=1$, define $P^{(1)}(U,V)$ by\\
\begin{algorithm}[H]
\For{$(u,v)\in E$}{
    $V[v] \gets V[v] + U[u] \pmod{q}$
}
\end{algorithm}

\noindent\textbf{Recursive program.} For $\ell\geq 2$, we set $r=r(\ell)=\lceil\log\ell\rceil$ and define $P^{(\ell)}(U,V,W_1,\dots,W_{r})$:\\
\begin{algorithm}[H]
    $P^{(\lceil\ell/2\rceil)}(U,W_{r},W_1,\dots,W_{r-1})$\qquad\tcp{Propagate walks of length $\lceil\ell/2\rceil$ from $U$ to $W_{r}$}
    $P^{(\lfloor\ell/2\rfloor)}(W_{r},V,W_1,\dots,W_{r-1})$\qquad\tcp{Propagate walks of length $\lfloor\ell/2\rfloor$ from $W_{r}$ to $V$}
    $\left(P^{(\lceil\ell/2\rceil)}\right)^{-1}(U,W_{r},W_1,\dots,W_{r-1})$\qquad\tcp{Uncompute register $W_{r}$}
\end{algorithm}

Both programs are reversible.
For $P^{(1)}$, the inverse is obtained by replacing ``$+$'' with ``$-$'' in the update rule.
For $\ell\ge 2$, the inverse of $P^{(\ell)}$ is obtained by replacing the middle call
$P^{(\lfloor \ell/2\rfloor)}$ with its inverse~$\bigl(P^{(\lfloor \ell/2\rfloor)}\bigr)^{-1}$, since the first and third lines already form a compute--uncompute pair.

\begin{claim}\label{clm:only-V-changes}
For every $\ell\ge 1$, executing $P^{(\ell)}(U,V,W_1,\dots,W_r)$ leaves all registers unchanged except those in block $V$.
\end{claim}
\begin{proof}
We argue by induction on $\ell$.
For $\ell=1$, the program $P^{(1)}(U,V)$ updates only registers in $V$ by definition.

Assume the claim holds for all lengths smaller than $\ell\geq 2$ and consider $P^{(\ell)}(U,V,W_1,\dots,W_r)$.
By the induction hypothesis, the first call
$P^{(\lceil \ell/2\rceil)}(U,W_r,W_1,\dots,W_{r-1})$ can modify only registers in $W_r$.
The second call
$P^{(\lfloor \ell/2\rfloor)}(W_r,V,W_1,\dots,W_{r-1})$ can modify only registers in $V$.
Finally, the third line applies the inverse of the first call, restoring $W_r$ to its previous contents.
Therefore, overall, only registers in $V$ may differ from their initial values.
\end{proof}

For $\ell\in\mathbb{N}^*$, $v\in[n]$, and $\tau=(\tau_1,\dots,\tau_n)\in\mathbb{Z}_q^n$, let $\alpha_{\ell,v}(\tau)\in\mathbb{Z}_q$ denote the value added to the register~$V[v]$ by executing~$P^{(\ell)}(U,V,W_1,\dots,W_r)$ when the initial contents of~$U$ are given by~$U[i]=\tau_i$ for all~$i\in[n]$ (and all other blocks are initialized by some fixed content).
We also write~$N_\ell(i,j)$ for the number of walks from~$i$ to~$j$ of length exactly~$\ell$, where by a walk we mean a sequence of neighboring vertices potentially containing loops.

\begin{claim}\label{claim: alpha-diff is number of walks}
For every $\ell\in\mathbb{N}^*$, $\tau,b\in\mathbb{Z}_q^n$, and $v\in[n]$,
\[
    \alpha_{\ell,v}(\tau+b)-\alpha_{\ell,v}(\tau)
    \equiv
    \sum_{i\in[n]} b_i\,N_{\ell}(i,v)
    \pmod{q}.
\]
In particular, if $b_s=1$ and $b_i=0$ for all $i\neq s$, then
\[
    \alpha_{\ell,t}(\tau+b)-\alpha_{\ell,t}(\tau)\equiv N_{\ell}(s,t)\pmod{q}.
\]
\end{claim}

\begin{proof}
We proceed by induction on $\ell$.

\smallskip
\noindent\emph{Base case $\ell=1$.}
For any $v\in[n]$ and $\tau\in\mathbb{Z}_q^n$, the program $P^{(1)}(U,V)$ adds to $V[v]$ the value
\[
    \alpha_{1,v}(\tau)\equiv \sum_{(i,v)\in E}\tau_i \pmod{q}.
\]
Hence, for any $\tau,b\in\mathbb{Z}_q^n$,
\[
    \alpha_{1,v}(\tau+b)-\alpha_{1,v}(\tau)
    \equiv
    \sum_{(i,v)\in E} b_i
    \pmod{q}.
\]
This equals $\sum_{i\in[n]} b_i\,N_1(i,v)\pmod{q}$, since $N_1(i,v)=1$ iff $(i,v)\in E$.

\smallskip
\noindent\emph{Inductive step.}
Assume the claim holds for all lengths smaller than $\ell\ge 2$.
Let $r=\lceil\log \ell\rceil$.
By Claim~\ref{clm:only-V-changes}, the first line of $P^{(\ell)}$ updates only the block $W_r$, and in particular transforms its contents from some initial vector~$\sigma\in\mathbb{Z}_q^n$ to
\[
    \sigma + \alpha_{\lceil \ell/2\rceil,\cdot}(\tau)\in\mathbb{Z}_q^n,
\]
where $\alpha_{\lceil \ell/2\rceil,\cdot}(\tau)$ denotes the vector whose $i$th entry is $\alpha_{\lceil \ell/2\rceil,i}(\tau)$.
The second line then applies $P^{(\lfloor \ell/2\rfloor)}(W_r,V,\dots)$, so the net value added to $V[v]$ is
\[
    \alpha_{\ell,v}(\tau)
    =
    \alpha_{\lfloor \ell/2\rfloor,v}\bigl(\sigma+\alpha_{\lceil \ell/2\rceil,\cdot}(\tau)\bigr),
\]
since the third line uncomputes $W_r$ without touching~$V$.

Therefore,
\begin{align*}
\alpha_{\ell,v}(\tau+b)-\alpha_{\ell,v}(\tau)
    &\equiv
    \alpha_{\lfloor \ell/2\rfloor,v}\bigl(\sigma+\alpha_{\lceil \ell/2\rceil,\cdot}(\tau+b)\bigr)
    -
    \alpha_{\lfloor \ell/2\rfloor,v}\bigl(\sigma+\alpha_{\lceil \ell/2\rceil,\cdot}(\tau)\bigr) \pmod{q}.
\end{align*}
Applying the induction hypothesis to the length-$\lfloor \ell/2\rfloor$ program (with input difference
$\alpha_{\lceil \ell/2\rceil,\cdot}(\tau+b)-\alpha_{\lceil \ell/2\rceil,\cdot}(\tau)$) gives
\begin{align*}
\alpha_{\ell,v}(\tau+b)-\alpha_{\ell,v}(\tau)
    &\equiv
    \sum_{i\in[n]}\bigl(\alpha_{\lceil \ell/2\rceil,i}(\tau+b)-\alpha_{\lceil \ell/2\rceil,i}(\tau)\bigr)\,N_{\lfloor \ell/2\rfloor}(i,v)
    \pmod{q}.
\end{align*}
Applying the induction hypothesis again, now to each term
$\alpha_{\lceil \ell/2\rceil,i}(\tau+b)-\alpha_{\lceil \ell/2\rceil,i}(\tau)$, yields
\begin{align*}
\alpha_{\ell,v}(\tau+b)-\alpha_{\ell,v}(\tau)
    &\equiv
    \sum_{i\in[n]}\Bigl(\sum_{j\in[n]} b_j\,N_{\lceil \ell/2\rceil}(j,i)\Bigr)\,N_{\lfloor \ell/2\rfloor}(i,v)
    \pmod{q}\\
    &\equiv
    \sum_{j\in[n]} b_j \Bigl(\sum_{i\in[n]} N_{\lceil \ell/2\rceil}(j,i)\,N_{\lfloor \ell/2\rfloor}(i,v)\Bigr)
    \pmod{q}.
\end{align*}
The inner sum counts length-$\ell$ walks from $j$ to $v$, by walk-concatenation $\sum_{i\in[n]} N_{\ell_1}(j,i)\,N_{\ell_2}(i,v) \;=\; N_{\ell_1+\ell_2}(j,v)$, and since $\lceil \ell/2\rceil+\lfloor \ell/2\rfloor=\ell$.
We obtain:
\[
    \alpha_{\ell,v}(\tau+b)-\alpha_{\ell,v}(\tau) \equiv \sum_{j\in[n]} b_j\,N_\ell(j,v) \pmod{q},
\]
which completes the induction.
\end{proof}

Claim~\ref{claim: alpha-diff is number of walks} immediately yields a simple catalytic procedure for counting length-$\ell$ $st$-walks modulo~$q$.
The idea is to run the propagation program twice: once with the initial contents of $U$ unchanged, and once after incrementing $U[s]$ by~$1$.
The difference in the resulting values of $V[t]$ isolates $N_\ell(s,t)\bmod q$.
We formalize this as follows.

\medskip
\noindent\textbf{Algorithm 1: Catalytic Savitch.}\\
\begin{algorithm}[H]
\For{$c\in\{0,1\}$}{
    $U[s] \gets U[s] + c \pmod{q}$\;
    $P^{(\ell)}(U,V,W_1,\dots,W_r)$\;
    $\alpha_{c} \gets V[t]$\;
    $\bigl(P^{(\ell)}\bigr)^{-1}(U,V,W_1,\dots,W_r)$\;
    $U[s] \gets U[s] - c \pmod{q}$\;
}
\Return{$\alpha_1-\alpha_0 \pmod{q}$}
\end{algorithm}

\begin{claim}\label{clm:catalytic-savitch-alg}
For all $\ell\in[n]$, Algorithm~1 runs in $\mathsf{CSPACE}(O(\log n),\,O(n\log^2 n))$ and outputs $N_{\ell}(s,t)\pmod{q}$.
\end{claim}

\begin{proof}
We first show correcteness of the algorithm and then analyze its complexity.

\smallskip
\emph{Correctness.}
Let $\tau\in\mathbb{Z}_q^n$ denote the initial contents of $U$, and $\tau_t'\in\mathbb{Z}_q$ the initial content of $V[t]$.
Further, let $b\in\mathbb{Z}_q^n$ be the vector with $b_s=1$ and $b_i=0$ for $i\neq s$.
By definition of $\alpha_{\ell,t}(\cdot)$, the two runs of~$P^{(\ell)}$ produce
\[
\alpha_0 \equiv \tau_t'+\alpha_{\ell,t}(\tau)\pmod{q}
\qquad\text{and}\qquad
\alpha_1 \equiv \tau_t'+\alpha_{\ell,t}(\tau+b)\pmod{q}.
\]
Therefore,
\[
\alpha_1-\alpha_0
\equiv
\alpha_{\ell,t}(\tau+b)-\alpha_{\ell,t}(\tau)
\equiv
N_\ell(s,t)
\pmod{q},
\]
where the last congruence is Claim~\ref{claim: alpha-diff is number of walks}.

\smallskip
\emph{Complexity.}
We already observed, the catalyst consists of $\lceil\log n\rceil+2$ blocks of $n$ registers over $\mathbb{Z}_q$, with~$\log q =O(\log n)$, for a total of $O(n\cdot\log^2 n)$ catalytic bits.
Further, Algorithm~1 restores the catalyst to its initial contents because each call to $P^{(\ell)}$ is followed by the inverse call $\bigl(P^{(\ell)}\bigr)^{-1}$, and we also undo the increment of $U[s]$.

It remains to bound the regular workspace.
Outside the calls to $P^{(\ell)}$ and its inverse, the algorithm stores~$\alpha_0,\alpha_1\in\mathbb{Z}_q$ and computes their difference, which requires $O(\log q)=O(\log n)$ bits.
To execute $P^{(\ell)}$ (and its inverse), we maintain an explicit control stack for the recursion.
At each recursion depth, there are only three stages to remember (corresponding to the three lines in the definition of $P^{(\ell)}$), so the stack can be represented by a sequence
$(s_1,\dots,s_{\lceil\log \ell\rceil})$ with $s_d\in\{1,2,3\}$.
This uses $O(\log \ell)=O(\log n)$ bits.
Given this control information, we can determine which blocks play the roles of the current input/output/intermediate blocks and iterate through the edge list to perform the required register updates.
Hence the total regular workspace is $O(\log n)$.
\end{proof}

We have shown how to compute $N_\ell(s,t)\bmod q$ for any prime modulus $q$ of $O(\log n)$ bits.
Ultimately, however, we want the exact value of $N_\ell(s,t)$, not just its residue modulo a single prime.
Fortunately, obtaining the exact value is no harder: by evaluating $N_\ell(s,t)$ modulo sufficiently many pairwise coprime moduli, we obtain its Chinese remainder representation, which can then be converted to binary in logspace using a result of Chiu, Davida and Litow~\cite{CDL01}.

Recall that if $m_1,\dots,m_k$ are pairwise coprime and $M:=\prod_{i=1}^k m_i$, then the Chinese remainder theorem implies that every integer $x\in\{0,1,\dots,M-1\}$ is uniquely determined by the residue vector
\[
(x\bmod m_1,\dots,x\bmod m_k).
\]
This residue vector is referred to as the \emph{Chinese remainder representation} of~$x$ (with respect to $m_1,\dots,m_k$).

We will use the following finding of~\cite{CDL01}.

\begin{lemma}[{\cite[Theorem~3.3]{CDL01}}]\label{lem: CRR to binary}
Conversion from Chinese remainder representation to binary is in~$\mathsf{NC}^1$.
\end{lemma}

Putting these ingredients together yields the catalytic endpoint of our trade-off.

\begin{theorem}\label{thm:warmup-stcon}
Given a directed graph $G$ of size $n$, vertices $s,t\in V(G)$, and a length $\ell\in[n]$, there is an algorithm running in $\mathsf{CSPACE}(O(\log n),\,O(n\log^2 n))$ that outputs $N_\ell(s,t)$.
In particular, $\mathrm{STCON}$ is contained in $\mathsf{CSPACE}(O(\log n),\,O(n\log^2 n))$.
\end{theorem}

\begin{proof}
For $\ell\in[n]$, the number of length-$\ell$ $st$-walks is at most $n^\ell\le n^n$.
Let $p_1,p_2,\dots$ denote the primes in increasing order, and set $M := \prod_{i=1}^{2n} p_i$.
We easily find that $N_\ell(s,t) \leq n^n\leq M$.
Therefore, $N_\ell(s,t)$ has a unique Chinese remainder representation with respect to the moduli $p_1,\dots,p_{2n}$.

We compute this representation by iterating over $p_1,\dots,p_{2n}$ and running Algorithm~1 with modulus~$q=p_i$ to obtain~$N_\ell(s,t)\bmod p_i$.
The primes $p_i$ can easily be generated in $O(\log n)$ space via iterating through~$O(\log n)$ bit integers and testing primality by checking divisibility by all smaller numbers using the space-efficient division algorithm from \cite{CDL01}.
Algorithm~1 restores the catalyst after each execution, so we may reuse the same catalytic tape across all moduli.
Thus, we obtain the full Chinese remainder representation of~$N_\ell(s,t)$ using~$O(\log n)$ regular workspace and~$O(n\cdot\log^2 n)$ catalytic memory.

Finally, by Lemma~\ref{lem: CRR to binary}, this Chinese remainder representation can be converted to the binary representation of $N_\ell(s,t)$ in $\mathsf{NC}^1$ and hence in particular in logarithmic space.

To decide $\mathrm{STCON}$ within the same bounds, it suffices to add a self-loop at $t$ and output whether $N_{n-1}(s,t)$ is nonzero.
\end{proof}

\section{Our Space-space Trade-off}\label{sec: space-space trade-off}

In this section we prove our main result, the space-space trade-off for directed $st$-connectivity.
As discussed in the introduction, the proof combines our catalytic Savitch variant with an idea from~\cite{BBRS98} based on partitioning the vertex set.
Rather than iterating over all candidate vertex sequences directly, their approach iterates over sequences of parts that a walk may cross, and for each such sequence checks whether there actually exists an $st$-walk consistent with it.
We adopt this viewpoint to ``outsource'' part of the bookkeeping from the catalyst of our Savitch variant to the regular workspace.

We partition the vertices $v\in[n]$ into $k$ groups of size roughly $\frac{n}{k}$ each, depending on their residue modulo~$k$. Fixing a residue $r$, we can easily order all vertices with that residue. Note that each vertex can be written as $v=\lambda_v\cdot k+r_v$ where $\lambda_v = \lfloor\frac{v}{k}\rfloor$ and $r_v\in\{0,...,k-1\}$. As such, $v$ can be seen as the $\lambda_v$-th vertex with residue $r_v$.

Every length-$\ell$ walk $(v_0,\dots,v_\ell)$ induces a residue sequence $(r_0,\dots,r_\ell)\in\{0,\dots,k-1\}^{\ell+1}$ via $r_i \equiv v_i \pmod{k}$.
We say that a residue sequence $(r_0,\dots,r_\ell)$ is \emph{consistent} with an $st$-walk of length $\ell$ if there exists a walk $(v_0,\dots,v_\ell)$ with $v_0=s$, $v_\ell=t$, and $v_i\equiv r_i\pmod{k}$ for all $i$.
On a high level, our algorithm enumerates all residue sequences in a Savitch-style recursion using
$O(\log \ell\cdot \log k)$ bits of regular workspace, and for each fixed residue sequence checks using $O(\frac{n}{k}\cdot \log^2n)$ catalytic memory whether it is consistent with an $st$-walk.

As before, the catalytic tape consists of $\lceil\log n\rceil+2$ blocks of registers $U,V,W_1,\dots,W_{\lceil \log n \rceil}$, but now each block contains only $\lceil n/k\rceil$ registers (rather than $n$), each storing an element of $\mathbb{Z}_q$ in $O(\log q)=O(\log n)$ bits, for a suitable prime modulus $q$.
Intuitively, a block stores values indexed by the \emph{position} $\lambda$ within a residue class, rather than by the vertex label itself.
This reduction in the number of registers per block is what yields the catalytic-memory savings.

For $\ell\in\mathbb{N}^*$ and residues $i,j\in\{0,\dots,k-1\}$, we define a family of reversible register programs $P^{(\ell)}_{i,j}$ which propagate length-$\ell$ walk-information from vertices with residue $i$ to vertices with residue $j$ modulo $k$.\vspace{.1cm}\\
\noindent\textbf{Base Program.} For $\ell=1$ and $i,j\in\{0,...,k-1\}$, we set $P^{(1)}_{i,j}(U,V)$:\\
    \begin{algorithm}[H]
    \For{$(u,v)\in E$ \textbf{with} $u \equiv i \pmod{k}$ \textbf{and} $v \equiv j \pmod{k}$}{
        $V[v] \gets V[v] + U[u] \pmod{q}$
    }
    \end{algorithm}
\noindent\textbf{Recursive Program.} For $\ell\geq 2$ and $i,j\in\{0,...,k-1\}$, we set $r=\lceil\log \ell\rceil$ and define $P^{(\ell)}_{i,j}(U,V,W_1,\dots,W_r)$:\\
    \begin{algorithm}[H]
    \For{$m\in\{0,...,k-1\}$}{
        $P_{i,m}^{(\lceil\ell/2\rceil)}(U,W_{r},W_1,\dots,W_{r-1})$\qquad\tcp{Propagate walks of length $\lceil\ell/2\rceil$ from $U$ to $W_{r}$}
        $P_{m,j}^{(\lfloor\ell/2\rfloor)}(W_{r},V,W_1,\dots,W_{r-1})$\qquad\tcp{Propagate walks of length $\lfloor\ell/2\rfloor$ from $W_r$ to $V$}
        $\left(P_{i,m}^{(\lceil\ell/2\rceil)}\right)^{-1}(U,W_{r},W_1,\dots,W_{r-1})$\qquad\tcp{Uncompute register $W_{r}$}
    }
    \end{algorithm}

We will sometimes write $P^{(\ell)}_{u,v}$ with $u,v\in[n]$ being vertices rather than residues in $\{0,\dots,k-1\}$.
By this we mean
\[
P^{(\ell)}_{u,v} \;:=\; P^{(\ell)}_{r_u,r_v},
\qquad\text{where } r_u\equiv u \pmod{k}\ \text{ and }\ r_v\equiv v \pmod{k}.
\]
As in the previous section, the programs $P^{(\ell)}_{i,j}$ are reversible, and their inverses are obtained by reversing the sequence of updates and replacing additions by subtractions.
Moreover, the same argument shows that only block $V$ is modified:

\begin{claim}\label{clm:only-V-changes-tradeoff}
For every $\ell\in\mathbb{N}^*$ and $i,j\in[n]$, executing $P^{(\ell)}_{i,j}(U,V,W_1,\dots,W_r)$ leaves all registers unchanged except those in block $V$.
\end{claim}

\medskip
Next we define the analogue of $\alpha_{\ell,v}(\cdot)$ for the residue-restricted programs.
Fix $\ell\in\mathbb{N}^*$ and $i,j\in[n]$.
For a vector $\tau=(\tau_0,\dots,\tau_{\lceil n/k\rceil-1})\in\mathbb{Z}_q^{\lceil n/k\rceil}$ and a vertex $v\in[n]$ with $v\equiv j \pmod{k}$, let $\alpha^{(i,j)}_{\ell,v}(\tau)\in\mathbb{Z}_q$ denote the value added to the register $V[\lambda_v]$ by executing
$P^{(\ell)}_{i,j}(U,V,W_1,\dots,W_r)$
when the input block $U$ is initialized so that for every vertex $u\equiv i\pmod{k}$, $U[\lambda_u]=\tau_{\lambda_u}$.

\begin{claim}\label{claim: alpha-ij-diff-is-number-of-walks}
For every $\ell\in\mathbb{N}^*$, $i,j\in[n]$, vectors
$\tau,b\in\mathbb{Z}_q^{\lceil n/k\rceil}$, and every vertex $v\in[n]$ with $v\equiv j\pmod{k}$,
\[
    \alpha^{(i,j)}_{\ell,v}(\tau+b)-\alpha^{(i,j)}_{\ell,v}(\tau)
    \equiv
    \sum_{\substack{u\in[n]:\\ u\equiv i\!\!\!\pmod{k}}}
    b_{\lambda_u}\,N_{\ell}(u,v)
    \pmod{q}.
\]
In particular, if $b_{\lambda_s}=1$ and $b_{\lambda}=0$ for all $\lambda\neq \lambda_s$, then
\[
    \alpha^{(s,t)}_{\ell,t}(\tau+b)-\alpha^{(s,t)}_{\ell,t}(\tau)
    \equiv
    N_{\ell}(s,t)
    \pmod{q}.
\]
\end{claim}

\begin{proof}
We prove the claim by induction on $\ell$.

\smallskip
\noindent\emph{Base case $\ell=1$.}
The argument is identical to the proof of Claim~\ref{claim: alpha-diff is number of walks}, restricted to edges $(u,v)\in E$ with $u\equiv i\pmod{k}$.

\smallskip
\noindent\emph{Inductive step.}
Assume the statement holds for all lengths smaller than $\ell\ge 2$ and let $r=\lceil\log \ell\rceil$.
By definition, $P^{(\ell)}_{i,j}$ loops over midpoint residues $m\in\{0,\dots,k-1\}$.
Fix such an $m$.
The first subcall~$P^{(\lceil \ell/2\rceil)}_{i,m}(U,W_r,W_1,\dots,W_{r-1})$
writes into~$W_r$ information corresponding to length-$\lceil \ell/2\rceil$ walks from residue~$i$ vertices to residue $m$ vertices.
The second subcall~$P^{(\lfloor \ell/2\rfloor)}_{m,j}(W_r,V,W_1,\dots,W_{r-1})$
then adds to~$V[\lambda_v]$ the contributions of length-$\lfloor \ell/2\rfloor$ walks from residue~$m$ vertices to residue $j$ vertices, using $W_r$ as its input block.
The third subcall uncomputes $W_r$ and does not affect $V$.

Formally, let $\sigma\in\mathbb{Z}_q^{\lceil n/k\rceil}$ denote the initial contents of $W_r$.
Then
\[
    \alpha^{(i,j)}_{\ell,v}(\tau) \equiv \sum_{m=0}^{k-1} \alpha^{(m,j)}_{\lfloor \ell/2\rfloor,v}\bigl(\sigma+\alpha^{(i,m)}_{\lceil \ell/2\rceil,\cdot}(\tau)\bigr) \pmod{q},
\]
where $\alpha^{(i,m)}_{\lceil \ell/2\rceil,\cdot}(\tau)$ is the vector over $\mathbb{Z}_q^{\lceil n/k\rceil}$ whose entry at position $\lambda_u$ (for $u\equiv m\pmod{k}$) equals $\alpha^{(i,m)}_{\lceil \ell/2\rceil,u}(\tau)$.

Using this, we can compute
\begin{align*}
    \alpha^{(i,j)}_{\ell,v}(\tau+b)-\alpha^{(i,j)}_{\ell,v}(\tau)
        &\equiv \sum_{m=0}^{k-1} \Bigl( \alpha^{(m,j)}_{\lfloor \ell/2\rfloor,v}\bigl(\sigma+\alpha^{(i,m)}_{\lceil \ell/2\rceil,\cdot}(\tau+b)\bigr)
            - \alpha^{(m,j)}_{\lfloor \ell/2\rfloor,v}\bigl(\sigma+\alpha^{(i,m)}_{\lceil \ell/2\rceil,\cdot}(\tau)\bigr)\Bigr)\pmod{q}.
\end{align*}
Applying the induction hypothesis to the length-$\lfloor \ell/2\rfloor$ programs gives
\begin{align*}
\alpha^{(i,j)}_{\ell,v}(\tau+b)-\alpha^{(i,j)}_{\ell,v}(\tau)
&\equiv
\sum_{m=0}^{k-1}
\sum_{\substack{u\in[n]:\\ u\equiv m\!\!\!\pmod{k}}}
\Bigl(
\alpha^{(i,m)}_{\lceil \ell/2\rceil,u}(\tau+b)-\alpha^{(i,m)}_{\lceil \ell/2\rceil,u}(\tau)
\Bigr)\,N_{\lfloor \ell/2\rfloor}(u,v)
\pmod{q}.
\end{align*}
Applying the induction hypothesis again to each difference
$\alpha^{(i,m)}_{\lceil \ell/2\rceil,u}(\tau+b)-\alpha^{(i,m)}_{\lceil \ell/2\rceil,u}(\tau)$ yields
\begin{align*}
\alpha^{(i,j)}_{\ell,v}(\tau+b)-\alpha^{(i,j)}_{\ell,v}(\tau)
&\equiv
\sum_{m=0}^{k-1}
\sum_{\substack{u\in[n]:\\ u\equiv m\!\!\!\pmod{k}}}
\left(
\sum_{\substack{w\in[n]:\\ w\equiv i\!\!\!\pmod{k}}}
b_{\lambda_w}\,N_{\lceil \ell/2\rceil}(w,u)
\right)
N_{\lfloor \ell/2\rfloor}(u,v)
\pmod{q}\\
&\equiv
\sum_{\substack{w\in[n]:\\ w\equiv i\!\!\!\pmod{k}}}
b_{\lambda_w}
\left(
\sum_{u\in[n]} N_{\lceil \ell/2\rceil}(w,u)\,N_{\lfloor \ell/2\rfloor}(u,v)
\right)
\pmod{q}.
\end{align*}
In the last step we used that the residue classes partition $[n]$, so
$\sum_{m}\sum_{u\equiv m}(\cdot)=\sum_{u\in[n]}(\cdot)$.
Finally, by walk concatenation, $\sum_{u\in[n]} N_{\ell_1}(w,u)\,N_{\ell_2}(u,v)=N_{\ell_1+\ell_2}(w,v)$,
and since $\lceil \ell/2\rceil+\lfloor \ell/2\rfloor=\ell$, we obtain
\[
\alpha^{(i,j)}_{\ell,v}(\tau+b)-\alpha^{(i,j)}_{\ell,v}(\tau)
\equiv
\sum_{\substack{w\in[n]:\\ w\equiv i\!\!\!\pmod{k}}}
b_{\lambda_w}\,N_\ell(w,v)
\pmod{q},
\]
as desired to finish the induction.
\end{proof}

The final algorithm for our space-space trade-off looks as follows:\vspace{.1cm}

\noindent\textbf{Algorithm 2: Catalytic Savitch for Space-space Trade-off.}\\
\begin{algorithm}[H]
    \For{$c\in\{0,1\}$}{
        $U[s] \gets U[s] + c$\\
        $P^{(\ell)}_{s,t}(U,V,W_1,\dots,W_r)$\\
        $\alpha_c\leftarrow V[t]$\\
        $\left(P^{(\ell)}_{s,t}\right)^{-1}(U,V,W_1,\dots,W_r)$\\
        $U[s] \gets U[s] - c$\\
    }
    \Return{$\alpha_1-\alpha_0 \pmod{q}$}\\
\end{algorithm}

We now state and prove the analogue of Claim~\ref{clm:catalytic-savitch-alg} for the residue-restricted programs from this section.

\begin{claim}
    For any $k\in [n]$ and $\ell\in[n]$, Algorithm $2$ runs in $\mathsf{CSPACE}(O(\log \ell\cdot\log k + \log n),O(\frac{n}{k}\cdot\log^2 n))$ and outputs $N_{\ell}(s,t) \pmod{q}$.
\end{claim}
\begin{proof}
The correctness proof is essentially identical to the one for \cref{clm:catalytic-savitch-alg}.

For the complexity, note that each block on the catalytic tape contains $\lceil n/k\rceil$ registers over $\mathbb{Z}_q$, with~$\log q =O(\log n)$, and we use $r+2=O(\log \ell)=O(\log n)$ such blocks.
Hence the total catalyst size indeed is~$O\left(\frac{n}{k}\cdot \log^2 n\right)$.

It remains to bound the regular workspace.
The only difference to the proof of Claim~\ref{clm:catalytic-savitch-alg} is the additional workspace overhead that is obtained by including the loop over midpoint residues $m\in\{0,\dots,k-1\}$.
Keeping track of the current value of $m$ requires $O(\log k)$ bits per recursion level.
Thus the full control information can be stored in $O(r\cdot \log k)=O(\log \ell\cdot \log k)$ bits of regular workspace.
Given this control information, we can determine which residue-restricted base updates to apply and iterate through the edge list accordingly in space $O(\log n)$.
Therefore, the total regular workspace is~$O(\log \ell\cdot \log k + \log n)$, as claimed.
\end{proof}

As in Section~\ref{sec: Catalytic Savitch}, we can recover the exact value of $N_\ell(s,t)$ by evaluating its residues for sufficiently many primes to obtain a Chinese remainder representation, and then convert it to binary using Lemma~\ref{lem: CRR to binary}.
Completely analogously, we obtain our main result from the introduction, which we restate here:
\STCONTradeoff*

\section*{Acknowledgements}
The author thanks Simon Apers, Bruno Loff, Benjamin Mathieu-Bloise, Edward Pyne and Quinten Tupker for encouraging and insightful discussions.
This work has received support under the program ``Investissement d'Avenir" launched by the French Government and implemented by ANR, with the reference ``ANR‐22‐CMAS-0001, QuanTEdu-France".

\bibliographystyle{alpha}
\bibliography{refs}

@InProceedings{Wig92,
    author="Wigderson, Avi",
    title="The complexity of graph connectivity",
    booktitle="Mathematical Foundations of Computer Science 1992",
    year="1992",
    publisher="Springer",
    pages="112--132",
}

@article{Sav70,
    author = {Savitch, Walter J.},
    title = {Relationships between nondeterministic and deterministic tape complexities},
    year = {1970},
    issue_date = {April, 1970},
    publisher = {Academic Press, Inc.},
    address = {USA},
    volume = {4},
    number = {2},
    issn = {0022-0000},
    url = {https://doi.org/10.1016/S0022-0000(70)80006-X},
    doi = {10.1016/S0022-0000(70)80006-X},
    journal = {J. Comput. Syst. Sci.},
    month = apr,
    pages = {177–192},
    numpages = {16}
}

@article{BBRS98,
    author = {Barnes, Greg and Buss, Jonathan F. and Ruzzo, Walter L. and Schieber, Baruch},
    title = {A Sublinear Space, Polynomial Time Algorithm for Directed s-t Connectivity},
    journal = {SIAM Journal on Computing},
    volume = {27},
    number = {5},
    pages = {1273-1282},
    year = {1998},
    doi = {10.1137/S0097539793283151}
}

@inproceedings{CP25,
    title={Efficient Catalytic Graph Algorithms},
    author={Cook, James and Pyne, Edward},
    booktitle={17th Innovations in Theoretical Computer Science Conference (ITCS 2026)},
    volume={362},
    pages={43:1--43:22},
    year={2026},
    organization={Schloss Dagstuhl--Leibniz-Zentrum f{\"u}r Informatik},
    doi={10.4230/LIPIcs.ITCS.2026.43}
}

@article{CDL01,
    author = {Chiu, Andrew and Davida, George and Litow, Bruce},
    title = {Division in logspace-uniform $\mbox{NC}^1$},
    journal = {RAIRO. Theoretical Informatics and Applications},
    pages = {259--275},
    year = {2001},
    publisher = {EDP-Sciences},
    volume = {35},
    number = {3},
    mrnumber = {1869217},
    zbl = {1014.68062},
    language = {en},
    url = {https://www.numdam.org/item/ITA_2001__35_3_259_0/}
}

@inproceedings{BCKLS14,
    author = {Buhrman, Harry and Cleve, Richard and Kouck\'{y}, Michal and Loff, Bruno and Speelman, Florian},
    title = {Computing with a full memory: catalytic space},
    year = {2014},
    isbn = {9781450327107},
    publisher = {Association for Computing Machinery},
    address = {New York, NY, USA},
    url = {https://doi.org/10.1145/2591796.2591874},
    doi = {10.1145/2591796.2591874},
    booktitle = {Proceedings of the Forty-Sixth Annual ACM Symposium on Theory of Computing},
    pages = {857–866},
    numpages = {10},
    location = {New York, New York},
    series = {STOC '14}
}

@inproceedings{CM24,
    author = {Cook, James and Mertz, Ian},
    title = {Tree Evaluation Is in Space $O(\log n \cdot \log \log n)$},
    year = {2024},
    isbn = {9798400703836},
    publisher = {Association for Computing Machinery},
    address = {New York, NY, USA},
    url = {https://doi.org/10.1145/3618260.3649664},
    doi = {10.1145/3618260.3649664},
    booktitle = {Proceedings of the 56th Annual ACM Symposium on Theory of Computing},
    pages = {1268–1278},
    numpages = {11},
    location = {Vancouver, BC, Canada},
    series = {STOC 2024}
}

@inproceedings{Will25,
    author = {Williams, R. Ryan},
    title = {Simulating Time with Square-Root Space},
    year = {2025},
    isbn = {9798400715105},
    publisher = {Association for Computing Machinery},
    address = {New York, NY, USA},
    url = {https://doi.org/10.1145/3717823.3718225},
    doi = {10.1145/3717823.3718225},
    booktitle = {Proceedings of the 57th Annual ACM Symposium on Theory of Computing},
    pages = {13–23},
    numpages = {11},
    location = {Prague, Czechia},
    series = {STOC '25}
}

@article{CMMBS12,
    author  = {Cook, Stephen A. and McKenzie, Pierre and Wehr, Dustin and Braverman, Mark and Santhanam, Rahul},
    title   = {Pebbles and Branching Programs for Tree Evaluation},
    journal = {ACM Transactions on Computation Theory},
    volume  = {3},
    number  = {2},
    year    = {2012},
    doi     = {10.1145/2077336.2077337}
}

@article{CDKMR26,
    title={Frontier Space-Time Algorithms Using Only Full Memory},
    author={Chmel, Petr and Dudeja, Aditi and Koucký, Michal and Mertz, Ian and Rajgopal, Ninad},
    journal={arXiv preprint},
    year={2026}
}

\appendix

\section{Dealing with Invalid Initial Registers}\label{appendix:A}
Suppose a register program uses $m$ registers $R_1,\dots,R_m$, each storing a value $r_i\in\mathbb{Z}_q$.
We can encode the entire register configuration as a single integer
\[
    x \;=\; r_1 + r_2\cdot q + \dots + r_m \cdot q^{m-1} \;\in\; \mathbb{Z}_{q^m}.
\]
This encoding uses only $\lceil m\log q\rceil$ bits on the catalytic tape and is bijective between $\mathbb{Z}_q^m$ and $\mathbb{Z}_{q^m}$.

As for other encodings, it may happen that the catalytic tape initially encodes an invalid value, in this case this would be an integer $x\geq q^m$.
However, since we store all registers in a single number, we can easily handle this case with just one bit of extra memory: First flip the most significant bit of the catalyst from~$1$ to~$0$, and remember that you did so with one bit of memory.
This flip forces the encoded value into the range~$[0,q^m)$.
Then run the computation with the new valid initial tape, and finally flip the bit back to~$1$.

Observe that using space-efficient division with remainder from~\cite{CDL01}, one can extract each~$r_i\in\mathbb{Z}_q$ from $x\in\mathbb{Z}_{q^m}$ by computing
\[
    r_i \;=\; \left\lfloor \frac{x}{q^{i-1}} \right\rfloor \bmod q.
\]

We now show how to simulate register updates in this encoding.

\begin{lemma}
Let $x$ encode $(r_1,\dots,r_m)$ as above.
There is a procedure using $O(\log q + \log m)$ additional workspace that transforms $x$ into the encoding of the updated register tuple obtained by performing either
\[
    R_i \leftarrow R_i \pm R_j \pmod{q}
    \qquad\text{or}\qquad
    R_i \leftarrow R_i \pm c \pmod{q},
\]
for a constant $c\in \mathbb{Z}_q$, while all other registers remain unchanged.
\end{lemma}

\begin{proof}
First, we extract the value of $r_i$, and also of $r_j$ if we are performing an update of the first type.
We store these values in the regular workspace using $O(\log q)$ bits.

We then update $x$ in-place on the catalytic tape via
\[
    x \leftarrow x \pm u\cdot q^{i-1} \pmod{q^m},
\]
where $u\in\{0,...,q-1\}$ denotes either the value of $r_j$ or the constant $c$.
Denote the resulting value by $x'$ and write its decomposition as
\[
    x' = r_1' + r_2' \cdot q + \dots + r_m'\cdot q^{m-1}.
\]
It is easy to verify that $r_i'$ has the correct value $r_i \pm u \pmod{q}$.

However, if $r_i + u \geq q$ or $r_i - u <0$, then the update produces a carry, and $r_{i+1}'$ will be $r_{i+1}\pm 1$.
Since we have access to both $r_i$ and $u$ on our workspace, we can easily detect this though, and correct the carry by updating $x$ in-place (if-needed) via
\[
    x \leftarrow x \mp q^{i} \pmod{q^m}.
\]

The net effect is a correct simulation of the register update $r_i \leftarrow r_i \pm u \pmod{q}$ on the joint encoding $x$, using $O(\log q + \log m)$ additional workspace.
\end{proof}

\end{document}